\documentclass[a4paper,UKenglish,numberwithinsect]{lipics}
\usepackage{paralist}
\usepackage[kerning]{microtype}

\theoremstyle{plain}
\newtheorem*{theorem*}{Theorem}
\newtheorem{proposition}[theorem]{Proposition}
\theoremstyle{definition}
\newtheorem{assumptions}[theorem]{Assumptions}
\newtheorem{rem}[theorem]{Remark}
\newtheorem{notation}[theorem]{Notation}

\usepackage[all,pdf,2cell]{xy}
\SelectTips {cm} {}
\newdir{ >}{{}*!/-5pt/@{>}}

\newcommand{\pullbackcorner}[1][dr]{\save*!/#1-1.2pc/#1:(-1,1)@^{|-}\restore}

\newcommand{\<}{\langle}
\renewcommand{\>}{\rangle}
\newcommand{\hyph}{\text{-}}
\newcommand{\op}{\mathit{op}}
\newcommand{\id}{\mathit{id}}

\mathchardef\ordinarycolon\mathcode`\:
\mathcode`\:=\string"8000
\begingroup \catcode`\:=\active
  \gdef:{\mathrel{\mathop\ordinarycolon}}
\endgroup
\newcommand{\defeq}{:=}

\newcommand{\epito}{\twoheadrightarrow}

\renewcommand{\hom}{\mathsf{Hom}}
\newcommand{\bihom}{\mathsf{Bihom}}
\newcommand{\Pro}{\mathsf{Pro}}
\newcommand{\Ind}{\mathsf{Ind}}
\newcommand{\set}[2]{\{\,#1\mid#2\,\}}

\newcommand{\Free}{\mathbf{Free}}
\DeclareMathOperator*{\Lim}{Lim}
\DeclareMathOperator*{\Colim}{Colim}

\newcommand{\A}{\mathscr{A}}
\newcommand{\B}{\mathscr{B}}
\newcommand{\C}{\mathscr{C}}
\newcommand{\D}{\mathscr{D}}
\newcommand{\E}{\mathscr{E}}

\newcommand{\CSet}{\mathbf{Set}}
\newcommand{\CPos}{\mathbf{Pos}}
\newcommand{\CBA}{\mathbf{BA}}
\newcommand{\CDLat}{\mathbf{DLat}}
\newcommand{\CSLat}{\mathbf{SLat}}
\newcommand{\CVec}{\mathbf{Vec}}

\newcommand{\CMon}{\mathbf{Mon}}
\newcommand{\CPFMon}{\Pro\CMon_f}
\newcommand{\CAut}{\mathbf{Aut}}

\newcommand{\FLang}{\mathbf{LAN}}
\newcommand{\FVar}{\mathbf{LPV}}
\newcommand{\FPFMon}{\mathbf{PFMon}}

\newcommand{\Quo}{\mathsf{Quo}_f} 

\newcommand{\Reg}{\mathsf{Reg}}

\title{A Fibrational Approach to Automata Theory}
\author{Liang-Ting~Chen}
\author{Henning~Urbat}
\affil{Institut f\"{u}r Theoretische Informatik \\
  Technische Universit\"{a}t Braunschweig, Germany}
\Copyright{Liang-Ting Chen and Henning Urbat}

\begin{document}
\maketitle
\begin{abstract}
  For predual categories $\C$ and $\D$ we establish isomorphisms
  between opfibrations representing local varieties of languages in
  $\C$, local pseudovarieties of $\D$-monoids, and finitely generated
  profinite $\D$-monoids. The global sections of these opfibrations are shown to correspond to varieties of
  languages in $\C$, pseudovarieties of $\D$-monoids, and profinite equational theories of
  $\D$-monoids, respectively. As an application, we obtain a new proof of Eilenberg's variety theorem along with several related results, covering varieties of languages and their coalgebraic modifications, Straubing's $\mathsf{C}$-varieties, fully
  invariant local varieties, etc., within a single framework.
\end{abstract}
\section{Introduction}
In algebraic automata theory, regular languages are studied in connection with associated algebraic structures, using Eilenberg's celebrated variety
theorem~\cite{Eilenberg1976}. This theorem establishes a one-to-one correspondence
between varieties of languages and pseudovarieties of monoids. By a
\emph{variety of languages} is meant a class of regular languages closed
under the boolean operations (union, intersection and complement), left and right derivatives, and preimages under free monoid morphisms. A \emph{pseudovariety of monoids} is a class of
finite monoids closed under submonoids, quotients, and finite products.

Not every interesting class of languages falls within this scope. For this
reason several authors weakened the closure properties in the definition of a
variety of languages, and proved Eilenberg-type theorems for these modified
varieties. For example, Pin's \emph{positive varieties}~\cite{Pin1995}, omitting
closure under complement, correspond to pseudovarieties of \emph{ordered}
monoids. Pol\'{a}k's \emph{disjunctive varieties}~\cite{Polak2001}, further
dropping closure under intersection, correspond to pseudovarieties of idempotent
semirings. Reutenauer's \emph{xor varieties}~\cite{Reutenauer1980}, closed under
symmetric difference in lieu of the boolean operations, correspond to
pseudovarieties of associative algebras over the field $\mathbb{Z}_2$.
Straubing~\cite{Straubing2002} introduced \emph{$\mathsf{C}$-varieties of
  languages}, where one restricts to closure under preimages of a chosen class
$\mathsf{C}$ of free monoid morphisms in lieu of all free monoid morphisms.
They are in bijection with \emph{$\mathsf{C}$-pseudovarieties of monoid morphisms}, these
being classes of monoid morphisms with suitable closure properties. 

A closely related line of work concerns ``local'' versions of Eilenberg's
variety theorem, where languages over a fixed alphabet $\Sigma$ are considered.
Using the well-known duality between boolean algebras and Stone spaces,
Pippenger~\cite{Pippenger2007} demonstrated that the boolean algebra
$\mathsf{Reg}(\Sigma)$ of all regular languages over $\Sigma$ dualises to the
underlying Stone space of the free profinite monoid on~$\Sigma$. Later, Gehrke,
Grigorieff, and Pin~\cite{Gehrke2008} considered  \emph{local varieties of
  languages over $\Sigma$}, i.e.~boolean subalgebras of $\mathsf{Reg}(\Sigma)$
closed under left and right derivatives, and characterised them as sets of
regular languages over $\Sigma$ definable by profinite equations. 

In the recent work of Ad\'{a}mek, Milius, Myers, and
Urbat~\cite{Adamek2014c,Adamek2015} a categorical approach to Eilenberg-type
theorems was presented, covering many of the aforementioned results uniformly.
The leading idea is to take two varieties of (possibly ordered) algebras $\C$
and $\D$ whose full subcategories of finite algebras are dually equivalent.
Local varieties of languages are then modelled as coalgebras in $\C$, and
monoids as monoid objects in $\D$. The main result of~\cite{Adamek2014c}, the
\textbf{General Local Variety Theorem}, states that \emph{local varieties of
  languages over $\Sigma$ in $\C$} (= sets of regular languages over $\Sigma$
closed under $\C$-algebraic operations and left and right derivatives)
correspond to \emph{local pseudovarieties of $\Sigma$-generated $\D$-monoids} (=
sets of $\Sigma$-generated finite $\D$-monoids closed under quotients and
subdirect products). The \textbf{General Variety Theorem} of~\cite{Adamek2015}
establishes a correspondence between \emph{varieties of languages in $\C$} and
\emph{pseudovarieties of $\D$-monoids}. Then the classical Eilenberg theorem is
recovered by taking $\C =$ boolean algebras and $\D=$ sets, and other choices of
$\C$ and $\D$ give its modifications due to Pin, Pol\'{a}k and Reutenauer along
with new concrete Eilenberg-type correspondences.

The present paper is a continuation of the above work, aiming at two intriguing
questions:
\begin{enumerate}
  \item the connection between local pseudovarieties of $\D$-monoids and
    profinite $\D$-monoids;
  \item the connection between the local and non-local versions of the General
    Variety Theorem; 
\end{enumerate}
left open in~\cite{Adamek2014c,Adamek2015}.
To attack these questions, we organise all local varieties of languages into a
category $\FLang$ whose objects are pairs $(\Sigma,V)$ of a finite alphabet
$\Sigma$ and a local variety of languages over $\Sigma$ in $\C$. With a suitable
choice of morphisms in $\FLang$ (see Definition~\ref{def:flang}) the projection
functor $p\colon\FLang\to\Free(\CMon\D)$ into the category of finitely generated
free $\D$-monoids, mapping $(\Sigma,V)$ to the free $\D$-monoid over $\Sigma$, is
an opfibration. In a similar fashion one can form the category $\FVar$ of local
pseudovarieties of $\D$-monoids and the category $\FPFMon$ of finitely generated
profinite $\D$-monoids, which again yield opfibrations over $\Free(\CMon\D)$.
\[
\xymatrix{
      \FLang\ar[dr]_{p} \ar[r]^\cong & \FVar \ar[d]^q \ar[r]^\cong & \FPFMon\ar[dl]^{q'} \\
      & \Free(\CMon\D) &
}
\]

Then we make two crucial observations. Firstly, we show that the global
sections (namely, right inverse functors) of the above opfibrations $p$, $q$ and $q'$
correspond precisely to varieties of languages in $\C$, pseudovarieties of
$\D$-monoids and profinite equational theories of $\D$-monoids, respectively.
Secondly, we prove that the three opfibrations are isomorphic. The isomorphism
$\FLang\cong\FVar$ is essentially the General Local Variety Theorem
of~\cite{Adamek2014c}, and the isomorphism $\FVar\cong\FPFMon$ is based on a
limit construction. From these isomorphisms it follows immediately that the
global sections of our three opfibrations are in bijective correspondence:
\begin{quote}
  \emph{There is a bijective correspondence between (i) varieties of languages
    in $\C$, (ii) pseudovarieties of $\D$-monoids and (iii) profinite equational
    theories of $\D$-monoids.}
\end{quote}
The bijection (ii)$\leftrightarrow$(iii) amounts to a categorical presentation
of the well-known Reiterman-Banaschewski
theorem~\cite{Reiterman1982,Banaschewski1983}. And (i)$\leftrightarrow$(ii)
gives a conceptually completely different categorical proof of the General
Variety
Theorem in~\cite{Adamek2015}. Furthermore, the flexibility of our fibrational
setting leads rather easily to a number of additional results. For example, by
replacing the category  $\Free(\CMon\D)$ with an arbitrary subcategory
$\mathsf{C}\hookrightarrow \Free(\CMon\D)$ we obtain a generalised version of
Straubing's variety theorem for $\mathsf{C}$-varieties of languages, as well as
a new local variety theorem for \emph{fully invariant local varieties of
  languages}, i.e.\ local varieties closed under preimages of endomorphisms of
free monoids.

Beyond these concrete results, we believe that the main contribution of the present paper is a further illumination of the intrinsic duality deeply hidden in  algebraic language theory, most notably of the subtle interweavings of local and non-local structures, and the role of profinite theories.

\section{Preliminaries}
\label{sec:review}

In this section we review the categorical approach to algebraic automata theory
developed in~\cite{Adamek2014c,Adamek2015}. The idea is to interpret local
varieties of languages inside a variety of algebras $\C$, and to relate them to
finite monoids in another variety of (possibly ordered) algebras $\D$ which is
\textbf{predual} to $\C$. The latter means that the full subcategories $\C_f$
and $\D_f$ of finite algebras are dually equivalent. Note that by an
\textbf{ordered algebra} we mean an algebra (over a finitary signature $\Gamma$)
with a poset structure on its underlying set making all operations  monotone.
Morphisms of ordered algebras are order-preserving $\Gamma$-homomorphisms. A
\emph{variety of ordered algebras} is a class of ordered algebras specified by
inequalities $t_1\leq t_2$ between $\Gamma$-terms. 

\begin{assumptions}\label{asm:basic}
  In the following $\C$ and $\D$ are predual varieties of algebras, where $\D$-algebras may be
    ordered, subject to the following conditions:
  \begin{enumerate}
    \item $\C$ and $\D$ are \textbf{locally finite}, i.e.\ every
      free algebra on a finite set is finite;
    \item epimorphisms in $\D$ are surjective;
    \item $\D$ is \textbf{entropic}, i.e.\ given an $m$-ary operation $\sigma$ and an $n$-ary operation $\tau$ in the signature of $\D$ and variables $x_{ij}$ ($i=1,\ldots ,m$, $j=1,\ldots , n$), the following equation holds in $\D$:
    \[
      \sigma(\tau(x_{11},\ldots ,x_{1n}),\ldots ,\tau(x_{m1},\ldots,x_{mn}))
    = \tau(\sigma(x_{11},\ldots , x_{m1}),\ldots , \sigma(x_{1n},\ldots,x_{mn})).
    \]
  \end{enumerate}
\end{assumptions}

\begin{notation} We write $\Phi\colon\CSet\to \C$ and $\Psi\colon\CSet\to \D$ for the left adjoints to the forgetful
  functors $|{-}|\colon\C\to\CSet$ and $|{-}|\colon\D\to\CSet$, respectively.  By $\mathbf{1}_\C = \Phi\mathbf{1}$ and
  $\mathbf{1}_\D=\Psi\mathbf{1}$ denote the free algebras over the singleton set.
\end{notation}

\begin{example}\label{ex:cd}
The following pairs of varieties $\C/\D$ satisfy our assumptions. The details of
  the first three examples can be found in~\cite{Johnstone1982}.
  \begin{enumerate}
      \item $\CBA/\CSet$: The Stone Representation Theorem exhibits a dual equivalence between the categories of finite
        boolean algebras and finite sets. It assigns to any finite boolean algebra $B$ the set $\CBA(B, \mathbf{2})$ of all homomorphisms into the two-chain $\mathbf{2}$. The dual
        of $h \colon A \to B$ is given by precomposition with $h$, i.e.\
        $f\in \CBA(B, \mathbf{2})$ is mapped to $f \circ h \in \CBA(A,
        \mathbf{2})$.
      \item $\CDLat/\CPos$: Similarly, the Birkhoff Representation Theorem exhibits a dual equivalence between the
        categories of finite distributive lattices with $0$ and $1$ and finite posets. It assigns to a
        finite distributive lattice $L$ the poset $\CDLat(L,
        \mathbf{2})$, ordered pointwise, where $\mathbf{2}$ is the two-chain. On morphisms the dual equivalence again acts by precomposition.
      \item $\CSLat/\CSLat$: The category of finite semilattices with $0$ is self-dual: the dual equivalence maps a finite semilattice $S$ to the semilattice $\CSLat_f(S,\mathbf{2})$ whose join is taken pointwise.
      \item $\mathbb{Z}_2\hyph\CVec/\mathbb{Z}_2\hyph\CVec$: 
        The category of finite-dimensional vector spaces
        over any field $F$ is self-dual, by mapping a vector space
        $V$ to its dual space $F\hyph\mathbf{Vec}(V, F)$. By restricting $F$ to the binary field
        $\mathbb{Z}_2$, the category is also locally finite.
    \end{enumerate}
\end{example}

\begin{rem}\label{rem:indpro}
Given a small finitely complete and cocomplete category $\A$ we denote by
$\mathcal{Y}\colon\A\to\Ind\A$ and $\mathcal{Y}^\op\colon\A\to\Pro\A$ the ind- and pro-completion of $\A$, i.e.\ the free completion under filtered colimits and cofiltered limits, respectively. There is an adjunction $F \dashv U \colon \Pro\A \to \Ind\A$ such
  that $\mathcal{Y}^\op = F \circ \mathcal{Y}$ and $\mathcal{Y} = U \circ
  \mathcal{Y}^\op$.
  \[
    \UseTwocells
    \xymatrix{
      & \A \ar[rd]^{\mathcal{Y}^\op} \ar[ld]_{\mathcal{Y}} \\
      \Ind\A \rrtwocell_U^F{'\bot} & & \Pro\A
    }
  \]
Applying this to $\A=\C_f$ with $\Ind(\C_f)= \C$ and $\Pro(\C_f) =
\Ind(\C_f^\op)^\op\cong \Ind(\D_f)^\op = \D^{op}$, we see that the equivalence
$\mathscr{C}_f\cong\mathscr{D}_f^\op$ extends to an adjunction between
$\mathscr{C}$ and $\mathscr{D}^\op$. We denote both the equivalence
$\C_f\cong\D_f^\op$ and the induced adjunction between $\C$ and $\D^\op$ by
  \[
    S \dashv P \colon\D_f^\op\xrightarrow{\cong}\C_f
    \quad\text{and}\quad
    S \dashv P \colon\D^\op\to\C.
  \]
\end{rem}

\subsection{Local varieties of languages in
  \texorpdfstring{$\boldsymbol{\C}$}{C}}
The coalgebraic treatment of automata roots in the observation that a
deterministic automaton without an initial state is a coalgebra $\gamma = \<
\gamma^\mathrm{1st}, \gamma^\mathrm{2nd}\>\colon Q \to \mathbf{2} \times
Q^\Sigma$ for the set functor $T_\Sigma^0 = \mathbf{2} \times (-)^\Sigma$. Here
$\Sigma$ is the finite input alphabet, $\mathbf{2} := \{\mathtt{yes},\,
\mathtt{no}\}$, $\gamma^\mathrm{1st}\colon Q\to \mathbf{2}$ is the
characteristic function of the final states, and $\gamma^\mathrm{2nd}\colon Q\to
Q^\Sigma$ is the transition map. In the following we consider automata in the
category $\C$, which requires to replace the set $\mathbf{2}$ by a suitable
``output'' object in $\C$.  Observe that the dual adjunction $S \dashv P
\colon\D^\op\to\C$ has \emph{dualising objects} $O_\C \defeq P\mathbf{1}_\D$ and
$O_\D \defeq S\mathbf{1}_\C$, that is, for all $M\in \D$ and $Q\in \C$ we have
\[
  |PM| \cong \C(\mathbf{1}_\C, PM)
  \cong \D(M, O_\D)
  \quad\text{and}\quad
  |SQ| \cong \D(\mathbf{1}_\D, SQ)\cong\C(Q, O_\C).
\]
Taking $M=\mathbf{1}_\D$ we see that the set $|O_\C|$ is isomorphic to $|O_\D|$. Note that in each of the categories $\C/\D$ in Example~\ref{ex:cd} the objects $O_\C$ and $O_\D$ have a two-element carrier. Motivated by this observation, we replace the set $\mathbf{2}$ by
the  object $O_\C$ to define automata in $\C$.

\begin{definition}
  A \textbf{$\boldsymbol{\Sigma}$-automaton in $\boldsymbol{\mathscr{C}}$} is a coalgebra $\gamma = \< \gamma^\mathrm{1st}, \gamma^\mathrm{2nd}\>\colon
    Q \to O_\C \times Q^\Sigma$
 for the endofunctor $T_\Sigma
  \defeq O_\C \times (-)^\Sigma$ on $\C$, where $(-)^\Sigma$ is the $\Sigma$-fold product.  A \textbf{subautomaton} of $(Q,\gamma)$ is a subcoalgebra of $(Q,\gamma)$, represented by an injective coalgebra homomorphism into $Q$. An automaton is called \textbf{finite} if the object $Q$ of states is finite, and \textbf{locally finite} if it is a filtered colimit of finite
  $\Sigma$-automata. The \textbf{rational fixpoint}~$\rho T_\Sigma$ is the
  filtered colimit of all finite $\Sigma$-automata.
  The categories of $\Sigma$-automata, finite $\Sigma$-automata and locally finite $\Sigma$-automata in $\C$ are denoted by $\CAut\Sigma$, $\CAut_f\Sigma$ and
  $\CAut_\mathit{lf}\Sigma$, respectively. Their morphisms are coalgebra homomorphisms.
\end{definition}

In~\cite{Milius2010, AMV2006} it is shown that the rational fixpoint $\rho
T_\Sigma$ is the terminal locally finite coalgebra (i.e.\ the terminal object of $\CAut_\mathit{lf}\Sigma$), with the
structure map $\rho T_\Sigma \xrightarrow{\zeta} T_\Sigma(\rho T_\Sigma)$ an
isomorphism. The rational fixpoint of the set functor $T_\Sigma^0 = \mathbf{2}\times (-)^\Sigma$ is the automaton of regular languages: the states of $\rho T_\Sigma^0$ form the set $\Reg(\Sigma)$ of regular languages over $\Sigma$, the final states are those languages containing the empty word $\varepsilon$, and the transitions are given by left derivatives, that is, $L\xrightarrow{a} a^{-1}L = \set{w\in\Sigma^*}{aw\in L}$ for $L\in\Reg(\Sigma)$ and $a\in\Sigma$. 

\begin{rem}\label{rem:ocod}
To simplify the presentation, we assume in the following that
$|O_\C|=|O_\D|=\mathbf{2}$. The main reason is that in this case the rational
fixpoint $\rho T_\Sigma$ is a lifting of the above automaton of regular
languages to $\C$, see the next proposition. Without this assumption one needs
to replace regular languages by regular behaviors, i.e.\ functions $\Sigma^*\to |O_\C|$ realised by finite Moore automata with output set $|O_\C|$. See also the discussion in~\cite[Section V]{Adamek2015}.
\end{rem}
\begin{proposition}[see~\cite{Adamek2014c}]
  \label{prop:characterisation_rhoT}
  The rational fixpoint $\rho T_\Sigma$ is carried by the set
  $\Reg(\Sigma)$. Its coalgebra structure $\rho T_\Sigma \xrightarrow{\zeta} O_\C \times (\rho T_\Sigma)^\Sigma$ is given by the $\C$-morphisms 
  \[
    \zeta^\mathrm{1st}(L) =
    \begin{cases}
      \mathtt{yes} & \text{if $\varepsilon\in L$;} \\
      \mathtt{no} & \text{otherwise,}
    \end{cases}
    \quad\text{and}\quad
    \zeta^\mathrm{2nd}(L)(a) = a^{-1}L.
  \]
 
\end{proposition}
In the light of this proposition we also write $\Reg(\Sigma)$ for the rational fixpoint $\rho T_\Sigma$.

\begin{example}
  For $\C = \CBA$, the rational fixpoint of $T_\Sigma$ is the boolean algebra $\Reg(\Sigma)$ (w.r.t $\cup$, $\cap$, $(-)^\complement$, $\emptyset$ and $\Sigma^*$), endowed with the
automata structure given by the boolean homomorphisms $\zeta^\mathrm{1st}$ and $\zeta^\mathrm{2nd}$. Similarly, for the other categories $\C$ of Example~\ref{ex:cd} the algebraic structure of $\rho T_\Sigma = \Reg(\Sigma)$ is
  \begin{inparaenum}[\itshape a\upshape)]
    \item $\cup$, $\cap$, $\emptyset$, and $\Sigma^*$ for $\C = \CDLat$;
    \item $\cup$ and $\emptyset$ for $\C = \CSLat$;
    \item symmetric difference $L\oplus L' = (L\setminus L')\cup (L'\setminus L)$ and $\emptyset$ for $\C = \mathbb{Z}_2$-$\CVec$.
  \end{inparaenum}
\end{example}

\begin{definition}
A \textbf{local variety of languages over $\boldsymbol{\Sigma}$ in $\boldsymbol{\C}$} is a
  subautomaton $V$ of $\rho T_\Sigma$ closed under right derivatives, i.e.\
  $L \in |V|$ implies $La^{-1}=\set{w\in\Sigma^*}{wa \in L}\in |V|$ for all $a \in
  \Sigma$.  The $\bigcap$-semilattices of all (finite) local
    varieties of languages over $\Sigma$ in $\C$ are denoted by $\FLang_\Sigma^f$ and $\FLang_\Sigma$, respectively.
\end{definition}

Observe that a local variety of languages is closed under%
\begin{inparaenum}[(i)]
  \item the $\C$-algebraic operations of $\rho T_\Sigma$, being a sub\emph{algebra} of $\rho T_\Sigma$ in $\C$, and
  \item left derivatives, being a sub\emph{coalgebra} of $\rho T_\Sigma$.  For
    $\C=\CDLat$ ($\C=\CBA$) a local variety of languages is precisely a
    \emph{(boolean) quotienting algebra of languages} in the sense of Gehrke
    et~al.~\cite{Gehrke2008}: a set of regular languages over $\Sigma$ closed
    under union, intersection (and complement)  as well as left and right
    derivatives.
\end{inparaenum}

\subsection{\texorpdfstring{$\boldsymbol{\D}$}{D}-monoids}

Every entropic variety~$\D$ of (ordered) algebras can be equipped
with a symmetric monoidal closed structure $(\D,\otimes,\mathbf{1}_\D)$, see~\cite{BN1976} and~\cite[Theorem 3.10.1]{Borceux1994}. The unit $\mathbf{1}_\D$ is the free one-generated algebra and $\otimes$ is the usual tensor product of algebras, giving rise to a natural bijection between morphisms
and bimorphisms in~$\D$:
\[
  \hom(A\otimes B, C) \cong 
  \bihom(A\times B, C).
\]
Recall that a \textbf{bimorphism} $f\colon A \times B \to C$ in $\D$ is a
set-theoretic function  from $A\times B$ to $C$ such that $f(a, -)\colon B\to C$
and $f(-, b)\colon A \to C$ are $\D$-morphisms for any $a \in A$ and $b
\in B$. 

Since the tensor product represents bimorphisms, the monoid objects of the monoidal category $(\D,\otimes,\mathbf{1}_\D)$ correspond to the following algebraic concept:

\begin{definition}
  A \textbf{$\boldsymbol{\D}$-monoid} $(M, \bullet, e)$ is an object $M$ of $\D$ equipped with a monoid structure $(|M|,
  \bullet, e)$ in $\CSet$ whose multiplication $\bullet\colon M\times M\to M$ is 
  a $\D$-bimorphism. By a \textbf{morphism} $f\colon (M,\bullet,e)\to
  (M',\bullet',e')$ of $\D$-monoids is meant a morphism $f\colon M\to M'$ of $\D$ that is also a monoid morphism between the underlying monoids in $\CSet$. By $\CMon_f \D$ and $\CMon\D$ we denote the categories of (finite) $\D$-monoids and all $\D$-monoid morphisms.
\end{definition}

\begin{example}\label{ex:monoids}
  For the categories $\D = \CSet$, $\CPos$, $\CSLat$ and $\mathbb{Z}_2$-$\CVec$
  of Example~\ref{ex:cd}, the $\D$-monoids
  are precisely ordinary monoids, ordered monoids, idempotent semirings (with $0$ and $1$) and associative algebras over the field $\mathbb{Z}_2$, respectively.
\end{example}

\begin{rem}
\begin{enumerate}
\item  In $\D$  we choose the factorisation system (epi, strong mono). Recall
  that epimorphisms in $\D$ are precisely the surjective morphisms by
  Assumption~\ref{asm:basic}.2. Strong monomorphisms are precisely the injective
  morphisms if $\D$ is a variety of algebras, and embeddings i.e.\ injective
  order-reflecting morphisms if $\D$ is a variety of ordered algebras. Hence
  every $\D$-morphism $f\colon A \to B$ factorises as
  $\xymatrix{A\ar@{->>}[r]^-{\mathrm{Im}(f)} & f[A]\ar@{ >->}[r]^i & B }$ where
  $\mathrm{Im}(f)$ is the restriction of $f$ to the image $f[A]$ and $i$ is
  injective (and order-reflecting).  Further, the factorisation system has the
  \textbf{fill-in property}: given a surjective morphism $e$, an injective (and
  order-reflecting) morphism $m$ and two morphisms $u, v$ with $ue = mv$, there
  is a unique morphism $d$ such that $u = md$ and $v = de$.

\item The factorisation system of $\D$ lifts to $\CMon\D$. Hence \textbf{submonoids} are represented by injective (order-reflecting) $\D$-monoid morphisms, and \textbf{quotient monoids} by surjective $\D$-monoid morphisms.
\end{enumerate}
\end{rem}


Since $\CMon\D$ is a variety of (ordered) algebras, the forgetful functor $\CMon\D\to\CSet$ has a left adjoint constructing free $\D$-monoids. Here is a concrete construction:

\begin{proposition}[see~\cite{Adamek2014c}]
The free $\D$-monoid on a set $\Sigma$ is carried by the $\D$-object $\Psi\Sigma^*$. The monoid multiplication $\bullet$ extends the concatenation of words in $\Sigma^*$, and the unit is $\epsilon$.
\end{proposition}

A \textbf{finite $\boldsymbol{\Sigma}$-generated $\boldsymbol{\D}$-monoid} is a
finite quotient $e_M\colon \Psi\Sigma^*\epito M$ of the free $\D$-monoid on
$\Sigma$. Given another finite $\Sigma$-generated $\D$-monoid $e_N\colon
\Psi\Sigma^*\twoheadrightarrow N$ we write $M\leq N$ if there is a $\D$-monoid
morphism $f\colon N \to M$ satisfying $e_M = fe_N$. With respect to this order all
(isomorphism classes of) finite $\Sigma$-generated $\D$-monoids form a poset
$\Quo(\Psi\Sigma^*)$. Observe that  $\Quo(\Psi\Sigma^*)$ is a join-semilattice:
the join of $M$ and $N$ is the \textbf{subdirect product}, viz.\ the image of the
morphism $\<e_M, e_N\>\colon \Psi\Sigma^*\to M\times N$ given by
\[
  M \vee N \defeq \set{ (e_M(x), e_N(x)) \in M \times N }{ x \in \Psi\Sigma^* }.
\]

\begin{definition}
  A \textbf{local pseudovariety of $\boldsymbol{\D}$-monoids over $\boldsymbol{\Sigma}$} is an
  \emph{ideal} of $\Quo(\Psi\Sigma^*)$, i.e.\ a set of finite $\Sigma$-generated
  $\D$-monoids closed under quotients and subdirect products. By $\FVar_\Sigma$
  we denote the $\bigcap$-semilattice of local pseudovarieties of $\D$-monoids over $\Sigma$.
\end{definition}

\begin{theorem}[{General Local Variety Theorem~\cite{Adamek2014c}}]\label{thm:genlocal}
  \label{thm:local-eilenberg-thm}
  For each finite alphabet $\Sigma$,
   \[\FLang^f_\Sigma \cong \Quo(\Psi\Sigma^*)\quad\text{and}\quad \FLang_\Sigma \cong \FVar_\Sigma.\]
\end{theorem}

\begin{rem}\label{rem:locvariso}
\begin{enumerate}
\item The first isomorphism takes a finite local variety $O_\C
  \xleftarrow{\gamma^\mathrm{1st}} V \xrightarrow{\gamma^\mathrm{2nd}} V^\Sigma$
  in $\C$ and applies the equivalence functor $S\colon \C_f\xrightarrow{\cong}
  \D_f^{op}$ to its coalgebra structure. This yields an algebra $\mathbf{1}_\D
  \cong S(O_\C) \xrightarrow{S\gamma^\mathrm{1st}} SV
  \xleftarrow{S\gamma^\mathrm{2nd}} S(V^\Sigma) \cong \coprod_\Sigma SV$ for the
  functor $F_\Sigma = \mathrm{1}_\D + \coprod_\Sigma(-)$ on $\D$. Since the free
  $\D$-monoid $\Psi\Sigma^*$ also carries the initial algebra for $F_\Sigma$,
  there is a unique $F_\Sigma$-algebra homomorphism $e_{SV}\colon
  \Psi\Sigma^*\to SV$ into the algebra constructed above. One then shows that
  $e_{SV}$ is surjective and there is a unique $\D$-monoid structure on $SV$
  making $e_{SV}$ a $\D$-monoid morphism.  We call $e_{SV}\colon
  \Psi\Sigma^*\epito SV$ the \textbf{(finite $\boldsymbol{\Sigma}$-generated)
    $\boldsymbol{\D}$-monoid corresponding to $\boldsymbol{V}$}.

\item The second isomorphism follows immedatiely from the observation that
  $\FLang_\Sigma$ is isomorphic to the ideal completion of $\FLang_\Sigma^f$.
  Indeed, every finite local variety of languages is a compact element of
  $\FLang_\Sigma$, and every local variety is the directed union of its finite
  local subvarieties. Hence the isomorphism $\FLang_\Sigma \cong \FVar_\Sigma$
  maps a local variety of languages $V\hookrightarrow \rho T_\Sigma$ to the
  local pseudovariety of all finite $\Sigma$-generated $\D$-monoids that
  correspond to some finite local subvariety of $V$. The inverse isomorphism
  maps a local pseudovariety $P$ of $\D$-monoids over $\Sigma$ to the directed
  union of all finite local varieties of languages in $\C$ that correspond to
  some element of $P$.
\end{enumerate}
\end{rem}

\subsection{Preimages under \texorpdfstring{$\boldsymbol{\D}$}{D}-monoid morphisms}
Recall from Remark~\ref{rem:ocod} that we assume $|O_\C|=|O_\D|=\mathbf{2}$.
Hence a language $L \subseteq \Delta^*$ may be identified with a morphism
$L\colon \Psi\Delta^*\to O_\D$ of $\D$, viz.\ the adjoint transpose of the
characteristic function $\Delta^* \to |O_\D|$.  Given this identification, the
\textbf{preimage} of $L$ under a $\D$-monoid morphism $f\colon \Psi\Sigma^* \to
\Psi\Delta^*$ is the composite $Lf\colon \Psi\Sigma^* \to\Psi\Delta^* \to O_\D$.
By the adjunction $S \dashv P \colon \D^\op\to\C$, the morphism $Pf$ is
essentially the preimage function, because \[|Pf| \cong \D(f, O_\D)\colon
  \D(\Psi\Delta^*,O_\D) \to \D(\Psi\Sigma^*,O_\D).\] In~\cite{Adamek2015} it was
shown that $|Pf|$ restricts to a $\C$-morphism $f^{-1}\colon \Reg(\Delta)\to
\Reg(\Sigma)$, taking any language $L\colon\Psi\Delta^*\to O_\D$ in
$\Reg(\Delta)$ to its $f$-preimage. This observation makes the following
definition evident:
\begin{definition}\label{def:preimg}
 Let $f\colon \Psi\Sigma^* \to \Psi\Delta^*$ be a $\D$-monoid morphism and $V$
 and~$W$ local varieties of languages over $\Sigma$ and $\Delta$, respectively.
 Then $V$ is said to be \textbf{closed under $\boldsymbol{f}$-preimages of
   languages in $\boldsymbol{W}$} if Diagram~\ref{eq:closed-under-preimages}
 below commutes for some $\C$-morphism $h$.
\end{definition}

\begin{figure}[h!]
  \renewcommand{\figurename}{Diagram}
  \centering
  \captionsetup{justification=centering}
  \begin{minipage}{0.3\textwidth}
  \[
     \vcenter{
      \xymatrix{
        W \ar@{ >->}[d]\ar[r]^{h} & V \ar@{ >->}[d]\\
        \Reg(\Delta) \ar[r]_{f^{-1}} & \Reg(\Sigma)
      }
    }
  \]
  \caption{}
  \label{eq:closed-under-preimages}
  \end{minipage}
  \begin{minipage}{0.05\textwidth}
  ~
  \end{minipage}
  \begin{minipage}{.3\textwidth}
  \[
     \vcenter{
      \xymatrix@C+1em{
        \Psi\Sigma^* \ar[r]^{f} \ar@{->>}[d]_{e_M} & \Psi\Delta^*
        \ar@{->>}[d]^{e_N} \\
        M \ar[r]_{g} & N
      }
    }
  \]
  \caption{}
  \label{eq:closed-under-preimage2}
  \end{minipage}
\end{figure}
Here is a dual characterisation of preimage closure:
\begin{lemma}[see~\cite{Adamek2015}]
  \label{lem:key-lemma}
  In Definition~\ref{def:preimg} let $V$ and $W$ be finite, and let $e_M\colon
  \Psi\Sigma^*\epito M$ and $e_N\colon \Psi\Delta^*\epito N$ be the finite $\D$-monoids corresponding to $V$ and $W$, respectively. Then  Diagram~\ref{eq:closed-under-preimages} commutes iff Diagram~\ref{eq:closed-under-preimage2} with $g=Sh$ commutes.
\end{lemma}
\section{Fibrations for Languages and Monoids}
\label{sec:duality}
We are ready to present our fibrational setting for (local) varieties of languages in $\C$ and (local) pseudovarieties of $\D$-monoids. For general information on fibred categories the reader is referred to~\cite{Jacobs1999}. Let us briefly recall some basic vocabulary:

\begin{definition}
  Let $p\colon \E \to\B$ be a functor.  
  \begin{enumerate}
  \item An object
    $X\in\E$ is \textbf{above} $I\in\B$ if $pX = I$, and
    similarly a morphism $f$ in~$\E$ is above a morphism~$u$
    in~$\B$ if $pf = u$. A morphism~$f$ is called \textbf{vertical}
    (over~$I$) if it is above an identity map (above $\id_I$).
    \item The \textbf{fibre} over $I\in\B$ is the subcategory $\E_I$ of $\E$ whose objects are the objects of $\E$ above $I$ and whose morphisms are the
      vertical morphisms over~$I$.
    \item A morphism $f\colon X\to Y$ of $\E$ is \textbf{opcartesian over}
      $u\colon I\to J$ in $\B$ if~$pf = u$ and for
      every morphism $g\colon X\to Z$ in~$\E$ above $wu$ for 
      $w\colon J\to pZ$, there is a unique morphism $h\colon Y\to Z$ above
      $w$ with $g=hf$.
    \item $p\colon\E\to\B$ is an \textbf{opfibration}
      over~$\B$ if
        for every $X \in \E$ and $u\colon pX\to J$ in $\B$ there is an
        opcartesian morphism $f\colon X\to Y$ above $u$, called an
        \textbf{opcartesian lifting} of $u$.
    \item Two opfibrations $p\colon \E\to\B$ and $p'\colon \E'\to\B$ are
      \textbf{isomorphic} is there is an isomorphism $i\colon \E\cong\E'$ preserving indices, that is, $p'i=p$. 
    \item A \textbf{global section} of an obfibration $p\colon \E\to\B$ is a
      functor $s\colon \B\to\E$ with $es=id$.
    \item A \textbf{poset opfibration} is an opfibration such that each
      fibre $\E_I$ ($I\in\B$) is a poset.
     \item A \textbf{$\boldsymbol{\B}$-indexed poset} is a functor
        $\mathcal{H}\colon\B\to\CPos$.
  \end{enumerate}
\end{definition}

All opfibrations we consider below are poset opfibrations. They are effectively interchangeable with indexed posets via the \textbf{Grothendieck construction}:
\begin{enumerate}
\item Given a poset opfibration $p\colon E\to\B$ one defines an
indexed poset $\mathcal{H}_p\colon\B\to\CPos$ as follows. Note first that every $\B$-morphism $I \xrightarrow{u} J$
with an
object~$X$ above~$I$ has a \emph{unique} opcartesian lifting $X
\xrightarrow{f} u^* X$ because $\E_J$ is a poset. Then $\mathcal{H}_p$ is defined by
\[
  I \mapsto \E_I \quad\text{and}\quad
  \left(I \xrightarrow{u} J\right) \mapsto \left(\E_I \xrightarrow{u^*}
    \E_J\right)
\]
where $u^*$ maps $X$ to $u^*X$.
\item Conversely, given an indexed poset
$\mathcal{H}\colon \B\to \CPos$, define the \textbf{Grothendieck completion} of
$\mathcal{H}$ to be the category $\int\mathcal{H}$ with
\begin{quote}
  \begin{description}
    \item[objects] $(I, x)$ where $I \in \B$ and $x \in
      \mathcal{H}I$;
    \item[morphisms] $(I, x) \xrightarrow{u} (J, y)$ where $I \xrightarrow{u} J$ is a
            morphism in~$\B$ with $\mathcal{H}u(x)
                  \leq_{\mathcal{H}J} y$.
  \end{description}
\end{quote}
Then the projection functor $p_\mathcal{H}\colon\int\mathcal{H}\to\B$ mapping $(I, x)$
to $I$ and $(I, x)\xrightarrow{u}(J, y)$ to $I\xrightarrow{u}J$ is an opfibration.
\end{enumerate}
The Grothendieck construction gives rise to an equivalence between suitable
$2$-categories of indexed posets and opfibrations. We only need the following
weaker statement:
\begin{theorem}[Grothendieck]
  \label{thm:equivalence-between-fibrations-and-indexed-categories}
 Every poset opfibration $p\colon\E\to\B$ is isomorphic to
 $p_{\mathcal{H}_p}\colon\E\to\B$, and every indexed poset $\mathcal{H}\colon
 \B\to\CPos$ is naturally isomorphic to $\mathcal{H}_{p_\mathcal{H}}\colon \B\to
 \CPos$. Furthermore, if $\mathcal{H}, \mathcal{H}'\colon \B\to \CPos$ are two
 naturally isomorphic indexed posets then $p_\mathcal{H},p_{\mathcal{H}'}$ are
 isomorphic opfibrations.
\end{theorem}

\subsection{Local pseudovarieties of \texorpdfstring{$\boldsymbol{\D}$}{D}-monoids as an opfibration}
In this section we organise the local pseudovarieties of $\D$-monoids into an opfibration $\FVar\to \Free(\CMon\D)$, or equivalently into an indexed poset $\Free(\CMon\D)\to \CPos$. The base category $\Free(\CMon\D)$ is the category of finitely generated free $\D$-monoids: its objects are finite sets $\Sigma$, and its morphisms $\Sigma\xrightarrow{f}\Delta$ are all $\D$-monoid morphisms $\Psi\Sigma^* \xrightarrow{f} \Psi\Delta^*$ between the free $\D$-monoids on $\Sigma$ and $\Delta$, respectively. Hence $\Free(\CMon\D)$ is dual to the Lawvere theory of the variety $\CMon\D$.

\begin{definition}\label{def:fvarindexed}
The indexed poset $(\mathord{-})_\sharp\colon \Free(\CMon\D) \to \CPos$ is defined as follows:
\begin{enumerate}
\item To each finite set $\Sigma$ it assigns the poset $\Sigma_\sharp =
  \FVar_\Sigma$ of all local pseudovarieties of $\D$-monoids over $\Sigma$,
  ordered by \emph{reverse} inclusion $\supseteq$.
\item To each $\D$-monoid morphism $f\colon \Psi\Sigma^*\to\Psi\Delta^*$ it
  assigns the monotone map $f_\sharp\colon \FVar_\Sigma\to\FVar_\Delta$, where
  for $P\in \FVar_\Sigma$ the local pseudovariety $f_\sharp(P)\in \FVar_\Delta$
  consists of all finite $\Delta$-generated $\D$-monoids $N$ with $e_Nf =ge_M$
  for some $M\in P$ and some morphism $g$; see Diagram~\ref{eq:closed-under-preimage2}.
\end{enumerate}
\end{definition}

\begin{lemma}\label{lem:hwelldef}
$(\mathord{-})_\sharp$ is a well-defined functor.
\end{lemma}

The Grothendieck construction applied to the indexed poset
$(\mathord{-})_\sharp\colon \Free(\CMon\D)\to \CPos$ yields the following equivalent opfibration:

\begin{definition}\label{def:fvar}
The category $\FVar$ of local pseudovarieties of $\D$-monoids has
\begin{quote}
  \begin{description}
    \item[objects] $(\Sigma, P)$ where $P$ is a local pseudovariety of $\D$-monoids over
      $\Sigma$;
    \item[morphisms] $(\Sigma, P) \xrightarrow{f} (\Delta, Q)$
      where $f\colon \Psi\Sigma^*\to\Psi\Delta^*$ is a $\D$-monoid morphism such
      that for every $N \in Q$  there exists $M \in P$  and $g\colon M \to N$
      subject to Diagram~\ref{eq:closed-under-preimage2}.
\end{description}
\end{quote}
The projection $\FVar \xrightarrow{q}\Free(\CMon\D)$ mapping $(\Sigma, P)$ to
$\Sigma$ and $(\Sigma, P) \xrightarrow{f} (\Delta, Q)$ to~$f$ is called the
\textbf{opfibration of local pseudovarieties of $\boldsymbol{\D}$-monoids}.

\end{definition}
\subsection{Local varieties of languages in
  \texorpdfstring{$\boldsymbol{\C}$}{C} as an opfibration}
In complete analogy to Definition~\ref{def:fvarindexed} and~\ref{def:fvar} we can define an indexed poset and its corresponding opfibration representing local varieties of languages in $\C$.

\begin{definition}\label{def:varindexed}
The indexed poset $(\mathord{-})_*\colon \Free(\CMon\D) \to \CPos$ is defined as follows:
\begin{enumerate}
\item To each finite set $\Sigma$ it assigns the poset $\Sigma_*=\FLang_\Sigma$ of all local varieties of languages over $\Sigma$ in $\C$, ordered by \emph{reverse} inclusion $\supseteq$.
\item To each $\D$-monoid morphism $f\colon \Psi\Sigma^*\to\Psi\Delta^*$ it assigns the monotone map
$f_*\colon \FLang_\Sigma\to\FLang_\Delta$, where for $V\in \FLang_\Sigma$ the local variety $f_*(V)\in \FLang_\Delta$ is the directed union of all local varieties $W$ satisfying Diagram~\ref{eq:closed-under-preimages} for some $h$. In other words, $f_*(V)$ is the \emph{largest} local variety of languages over $\Delta$ such that $V$ is closed under $f$-preimages of languages in~$f_*(V)$.
\end{enumerate}
\end{definition}

The Grothendieck construction gives the following opfibration:

\begin{definition}\label{def:flang}
  The category $\FLang$ of local varieties of languages in $\C$ has
  \begin{quote}
  \begin{description}
    \item[objects] $(\Sigma, V)$ where $V$ is a local variety of languages over $\Sigma$ in $\C$;
    \item[morphisms] $(\Sigma, V)\xrightarrow{f}(\Delta, W)$ 
      where $f\colon\Psi\Sigma^*\to\Psi\Delta^*$ is a $\D$-monoid morphism such that
      $V$ is closed under $f$-preimages of languages in~$W$.
  \end{description}
  \end{quote}
  The projection $\FLang \xrightarrow{p} \Free(\CMon\D)$ mapping $(\Sigma,
  V)$ to $\Sigma$ and $(\Sigma, V) \xrightarrow{f} (\Delta, W)$ to $f$ is called
  the \textbf{opfibration of local varieties of languages in $\boldsymbol{\C}$}.
\end{definition}

The General Local Variety Theorem (see Theorem~\ref{thm:genlocal}) implies that
the two indexed posets $(\mathord{-})_\sharp,(\mathord{-})_*\colon
\Free(\CMon\D)\to \CPos$ of Definition~\ref{def:fvarindexed}
and~\ref{def:varindexed} are naturally isomorphic. Indeed, recall from
Remark~\ref{rem:locvariso} that the isomorphism $\FVar_\Sigma\cong
\FLang_\Sigma$ sends a local pseudovariety $P\in\FVar_\Sigma$ to the directed
union of all finite
local varieties of languages over $\Sigma$ in $\C$ corresponding to the finite
$\Sigma$-generated $\D$-monoids in $P$. From this and Lemma~\ref{lem:key-lemma}
we conclude that the diagram below commutes for all $\D$-monoid morphisms
$f\colon\Psi\Sigma^*\to \Psi\Delta^*$.
\[\xymatrix{
      \FVar_\Sigma \ar[d]_{f_\sharp} \ar[r]^{\cong} & \FLang_\Sigma
      \ar[d]^{f_*} \\
      \FVar_\Delta \ar[r]_{\cong} & \FLang_\Delta
    }
\]
Hence, by Theorem~\ref{thm:equivalence-between-fibrations-and-indexed-categories}, we get an isomorphism between the corresponding opfibrations:
\begin{theorem}
  \label{thm:fibrational-eilenberg-thm}
  The opfibrations $p\colon \FLang\to \Free(\CMon\D)$ and $q\colon \FVar\to\Free(\CMon\D)$ are isomorphic. 
\end{theorem}

\begin{definition}
By a \textbf{variety of languages in $\boldsymbol{\C}$} is meant a global
section of $p$, i.e.\ a functor $\mathcal{V}\colon \Free(\CMon\D)\to \FLang$ with $p\mathcal{V}=\id$.
\end{definition}

In more concrete terms, a variety of languages in $\C$ is given by a
collection of local varieties $V_\Sigma\in \FLang_\Sigma$ (where $\Sigma$ ranges over all finite alphabets) such that for every  $f\colon \Psi\Sigma^*\to\Psi\Delta^*$ the local variety $V_\Sigma$ is closed under $f$-preimages of languages in
$V_\Delta$. Varieties of languages in the categories $\C=\CBA$, $\CDLat$, $\CSLat$ and $\mathbb{Z}_2\hyph\CVec$ of Example~\ref{ex:cd} are precisely the classical varieties of languages of Eilenberg~\cite{Eilenberg1976}, the positive varieties of Pin~\cite{Pin1995}, the disjunctive varieties of Pol\'ak~\cite{Polak2001} and the xor varieties of Reutenauer~\cite{Reutenauer1980}, respectively.

By Theorem~\ref{thm:fibrational-eilenberg-thm} every global section of $p\colon \FLang\to \Free(\CMon\D)$  corresponds uniquely to a
global section of $q\colon \FVar \to \Free(\CMon\D)$. In the next section we will see that also the global sections of $q$ admit a concrete interpretation.

\section{Profinite \texorpdfstring{$\boldsymbol{\D}$}{D}-Monoids}
\label{sec:profinite-monoids}
A \textbf{profinite $\boldsymbol{\D}$-monoid} is a cofiltered limit of finite
$\D$-monoids, and the \textbf{profinite completion}~$\widehat{M}$ of a
$\D$-monoid~$M$ is the cofiltered limit of the diagram of all its finite quotients. Since
limits in $\CMon\D$ are formed on the level of $\CSet$,
every profinite $\D$-monoid is equipped with a profinite topology,
i.e.\ it can be viewed as a Stone space if $\D$ is a variety of algebras
(or an ordered Stone space, if $\D$ is a variety of ordered algebras).%
\footnote{An \textbf{(ordered) Stone space} is a compact space such that for
  every $x \neq y$ (resp.\ $x \not\leq y$) there exists a clopen (upper) set
  containing $x$ but not $y$.}
By $\CPFMon\D$ denote the category of profinite $\D$-monoids with 
continuous (order-preserving) $\D$-monoid morphisms.
\begin{theorem}\label{thm:promon}
  \begin{enumerate}
    \item $\CPFMon\D$ is the pro-completion of the category $\CMon_f\D$ of
      finite $\D$-monoids (cf.\ Remark~\ref{rem:indpro}).
    \item The profinite completion $M\mapsto \widehat{M}$ gives a left adjoint to the forgetful functor
      $\CPFMon\D\to \CMon\D$.
  \end{enumerate}
\end{theorem}
The first item follows from~\cite[Proposition VI.2.4]{Johnstone1982}.  The
argument given there for varieties of algebras also applies to ordered algebras.
The second item follows from a standard argument for ordinary monoids, see
e.g.,~\cite[Theorem~3.2.7]{Rhodes2009}.

\begin{example}For our predual categories $\C/\D$ of Example~\ref{ex:cd} we obtain the following  descriptions of the categories $\Pro\D_f$, $\CMon\D$ and $\CPFMon\D$, cf.~\cite[Corollary VI.2.4]{Johnstone1982}.

\begin{table}[h!]
  \begin{tabular}{ l | l | l | l | l }
    $\C$ & $\D$ & 
    $\Pro\D_f$ & $\CMon\D$ & $\CPFMon\D$ \\
    \hline
    $\CBA$ & $\CSet$ & $\mathbf{Stone}$ & $\CMon$ & $\mathbf{Stone}(\CMon)$ \\
    $\CDLat$ & $\CPos$ & $\mathbf{OStone}$  & $\mathbf{OMon}$ & (to be
    characterised)
    \\
    $\mathbf{SLat}$ & $\mathbf{SLat}$ &  $\mathbf{Stone}(\mathbf{SLat})$ &
    $\mathbf{ISRing}$& $\mathbf{Stone}(\mathbf{ISRing})$ \\
    $\mathbb{Z}_2\hyph\mathbf{Vec}$ & $\mathbb{Z}_2\hyph\mathbf{Vec}$ 
    & $\mathbf{Stone}(\mathbb{Z}_2\hyph\mathbf{Vec})$ &
    $\mathbb{Z}_2\hyph\mathbf{Alg}$ &
    $\mathbf{Stone}(\mathbb{Z}_2\hyph\mathbf{Alg})$ 
  \end{tabular}
  \label{tab:categories}
\end{table}

$\mathbf{Stone}$ and $\mathbf{OStone}$ are the categories of (ordered) Stone spaces and continuous (order-preserving) maps.
The categories in the fourth column are the categories of monoids, ordered monoids, idempotent semirings and $\mathbb{Z}_2$-algebras, respectively; see Example~\ref{ex:monoids}.
By $\mathbf{Stone}(\A)$ for
a variety of algebras $\A$ we mean the category of $\A$-algebras
in $\mathbf{Stone}$.  For example, $\mathbf{Stone}(\CMon)$ 
is the category of monoids equipped with a Stone
topology (making the monoid multiplication continuous) and continuous monoid morphisms.
\end{example}

\subsection{Local pseudovarieties of
  \texorpdfstring{$\boldsymbol{\D}$}{D}-monoids vs. profinite
  \texorpdfstring{$\boldsymbol{\D}$}{D}-monoids}
In this section we show how to identify local pseudovarieties of $\D$-monoids
over $\Sigma$ with $\Sigma$-generated profinite $\D$-monoids. In the following
\textbf{quotients} of profinite $\D$-monoids are meant to be represented by
surjective continuous $\D$-monoid morphisms.  A
\textbf{$\boldsymbol{\Sigma}$-generated profinite $\boldsymbol{\D}$-monoid} is a
quotient of $\widehat{\Psi\Sigma^*}$, the profinite completion of the free
$\D$-monoid $\Psi\Sigma^*$. Note that, by Theorem~\ref{thm:promon},
$\widehat{\Psi\Sigma^*}$ is the free profinite $\D$-monoid on the free
$\D$-monoid $\Psi\Sigma^*$ w.r.t.\ the forgetful functor $\Pro\CMon_f\D\to
\CMon\D$, and hence also the free profinite $\D$-monoid on the set $\Sigma$
w.r.t.\ the composite forgetful functor $\Pro\CMon_f\D\to \CMon\D\to \CSet$.
The following standard facts will be useful.
\begin{lemma}[{see e.g.,~\cite[Chapter~3]{Rhodes2009}}]
  \label{lem:facts-inverse-systems}
  Let $F\colon\mathscr{J}\to\mathbf{KHaus}$ be a cofiltered diagram in the
  category of compact Hausdorff spaces and continuous functions. 
  \begin{enumerate}
    \item If every $F_i \xrightarrow{Ff} F_j$ for $i\xrightarrow{f}j$ is
      surjective, then the limit projections $\Lim F \xrightarrow{\pi_i} F_i$
      are also surjective. 
    \item If $\varphi\colon \Delta X \Rightarrow F$ is a cone over $F$ such that
      every projection $\varphi_i\colon X\to F_i$ is surjective, then the mediating morphism $X \to \Lim
      F$ is also surjective.
  \end{enumerate}
\end{lemma}

\begin{rem}\label{rem:pvartosmon}
\begin{enumerate}
\item To each local pseudovariety $P \in \FVar_\Sigma$ we associate a
  $\Sigma$-generated profinite $\D$-monoid as follows. Note first that $P$
  defines a cofiltered diagram in $\Pro\CMon\D$ via the projection $(e\colon
  \Psi\Sigma^* \epito M) \mapsto M$. Since the connecting morphisms are
  surjective, the above lemma implies that every limit projection $\Lim P \to M$
  for $M \in P$ is surjective. Moreover, given $P \subseteq P'$ in
  $\FVar_\Sigma$, there is a surjective mediating morphism $h\colon \Lim P' \to
  \Lim P$. In particular, taking $P'$ to be the local pseudovariety of
  \emph{all} finite quotients of $\Psi\Sigma^*$ with $\Lim
  P'=\widehat{\Psi\Sigma^*}$ we get a surjective morphism
  $\widehat{\Psi\Sigma^*}\epito \Lim P$, i.e.\ a $\Sigma$-generated profinite
  $\D$-monoid.

\item Conversely, to each $\Sigma$-generated profinite $\D$-monoid
  $e_\Sigma\colon \Psi\Sigma^*\epito F\Sigma$ we associate a local pseudovariety
  $\mathcal{V}_{F\Sigma}\in \FVar_\Sigma$ as follows: $\mathcal{V}_{F\Sigma}$
  consists of all finite $\Sigma$-generated $\D$-monoids of the form
  $\xymatrix{\Psi\Sigma^* \ar[r]^\eta & \widehat{\Psi\Sigma^*}
    \ar@{->>}[r]^{e_\Sigma} & F\Sigma \ar@{->>}[r]^{e_M} & M}$, where $\eta$ is
  the universal arrow of the adjunction between $\Pro\CMon_f \D$ and $\CMon\D$
  (see Theorem~\ref{thm:promon}) and $M$ is any finite quotient of $F\Sigma$.
  Observe that such a composite $e_Me_\Sigma\eta$ is always surjective: since
  $\widehat{\Psi\Sigma^*}$ is the limit of all finite quotients of
  $\Psi\Sigma^*$, and $M$ is finite (hence a finitely copresentable object of
  $\Pro\CMon\D$), the morphism $e_Me_\Sigma$ factorises through some limit
  projection $\pi_N$, where $N$ is a finite quotient of $\Psi\Sigma^*$:
   \[
     \xymatrix{
       \Psi\Sigma^* \ar[r]^\eta \ar@{->>}[rd] & \widehat{\Psi\Sigma^*} \ar@{->>}[d]_{\pi_N}
       \ar@{->>}[r]^{e_\Sigma} & F\Sigma\ar@{->>}[d]^{e_M} \\ & N \ar@{->>}[r]_f & M
     }
   \]
\end{enumerate}
\end{rem}

It is not difficult to to see that the two constructions of
Remark~\ref{rem:pvartosmon} are mutually inverse. More precisely:

\begin{theorem}
  \label{thm:identify-pseudovariety}
  Let $\Sigma$ be a finite set.
  \begin{enumerate}
    \item Every $\Sigma$-generated profinite $\D$-monoid~$F\Sigma$ corresponds
      uniquely to a local pseudovariety $\mathscr{V}_{F\Sigma}$ of
            $\D$-monoids over~$\Sigma$. That is,
      \[
        \mathsf{Quo}(\widehat{\Psi\Sigma^*})
        \cong \FVar_\Sigma,
      \]
      where $\mathsf{Quo}(\widehat{\Psi\Sigma^*})$ denotes the poset of $\Sigma$-generated profinite $\D$-monoids.
    \item Let $f\colon\Psi\Sigma^*\to\Psi\Delta^*$ be a $\D$-monoid morphism, $F\Sigma$ a $\Sigma$-generated profinite $\D$-monoid and $F\Delta$ a $\Delta$-generated profinite $\D$-monoid. Then the right-hand diagram below commutes for some $h$ iff 
    for every $N \in
      \mathscr{V}_{F\Delta}$ there is some $M\in\mathscr{V}_{F\Sigma}$ and a
      morphism $h_N$ making the left-hand diagram commute: 
      \[
        \vcenter{
          \xymatrix{
            \Psi\Sigma^* \ar[r]^{f} \ar@{->>}[d] & \Psi\Delta^*\ar@{->>}[d] \\
            M \ar[r]_{h_N} & N
          }
        }
        \qquad\quad
        \vcenter{
          \xymatrix{
            \widehat{\Psi\Sigma^*}\ar[r]^{\widehat{f}}\ar@{->>}[d]&
            \widehat{\Psi\Delta^*}\ar@{->>}[d]\\
            F\Sigma\ar[r]_{h} & F\Delta
          }
        }
      \]
  \end{enumerate}
\end{theorem}

From the opfibration  $q\colon \FVar\to\Free(\CMon\D)$ we thus get the following isomorphic opfibration:

\begin{definition}
  The category $\FPFMon$ has
\begin{quote}
  \begin{description}
    \item[objects] $(\Sigma, F\Sigma)$ where $F\Sigma$ is a $\Sigma$-generated profinite
      $\D$-monoid;
    \item[morphisms] $(\Sigma, F\Sigma) \xrightarrow{f} (\Delta, F\Delta)$
      where $f\colon \Psi\Sigma^*\to\Psi\Delta^*$ is a $\D$-monoid morphism making
      the following diagram commute for some $h$:
      \begin{equation}
        \vcenter{
        \label{eq:profinite-equational-theory}
        \xymatrix{
          \widehat{\Psi\Sigma^*}\ar[r]^{\widehat{f}}\ar@{->>}[d]&
          \widehat{\Psi\Delta^*}\ar@{->>}[d]\\
          F\Sigma\ar[r]_{h} & F\Delta
        }
      }
      \end{equation}
  \end{description}
  \end{quote}
   The projection $\FPFMon\xrightarrow{q'}\Free(\CMon\D)$ sending $(\Sigma,
   F\Sigma)$ to $\Sigma$ and $(\Sigma, F\Sigma)\xrightarrow{f}(\Delta, F\Delta)$
   to~$f$ is called the  \textbf{opfibration of finitely generated profinite
     $\boldsymbol{\D}$-monoids}.
\end{definition}
For the record:

\begin{corollary}\label{cor:isoqq}
The opfibrations $q\colon \FVar\to\Free(\CMon\D)$ and $q'\colon \FPFMon\to \Free(\CMon\D)$ are isomorphic.
\end{corollary}

%

\subsection{Pseudovarieties of \texorpdfstring{$\boldsymbol{\D}$}{D}-monoids
  vs.\ profinite equational theories}

By a \textbf{pseudovariety of $\boldsymbol{\D}$-monoids} is meant a class of finite $\D$-monoids closed under submonoids, quotients and finite products. In this section we relate pseudovarieties of $\D$-monoids to profinite equational theories of $\D$-monoids. 

\begin{definition}
  A \textbf{profinite equational theory} of $\D$-monoids is a global section
  $\mathcal{T}\colon \Free(\CMon\D)\to \FPFMon$ of the opfibration $q'\colon \FPFMon \to \Free(\CMon\D)$.
\end{definition}

More explicitly, a profinite equational theory associates to each finite set
$\Sigma$ a $\Sigma$-generated profinite monoid $e_\Sigma\colon
\widehat{\Psi\Sigma^*}\epito F\Sigma$ such that, for all $f\colon \Psi\Sigma^*\to\Psi\Delta^*$, diagram \eqref{eq:profinite-equational-theory} commutes for some $h$.

\begin{rem}\label{rem:pvartotheo}
\begin{enumerate}
\item To each profinite equational theory $\mathcal{T}$ with $\mathcal{T}\Sigma
  = (\Sigma, F\Sigma)$ we associate a pseudovariety $\mathscr{V}$ of
  $\D$-monoids as follows: $\mathscr{V}$ consists of all finite $\D$-monoids $M$
  such that for all $\D$-monoid morphisms $f\colon \widehat{\Psi\Sigma^*}\to M$
  there exists a (necessarily unique) $\D$-monoid morphism $\overline f\colon F\Sigma\to M$ with $\overline{f}e_\Sigma = f$.
\[
  \xymatrix@+.5em{
    {\widehat{\Psi\Sigma^*}} \ar[rd]_{f}\ar@{->>}[r]^{e_\Sigma} & F\Sigma
    \ar@{-->}[d]^{\overline{f}} \\
    & M
  }
\]
\item Conversely, to each pseudovariety $\mathscr{V}$ of $\D$-monoids we
  associate a profinite equational theory $\mathcal{T}$ with $\mathcal{T}\Sigma
  = (\Sigma, F\Sigma)$ as follows: given $\Sigma$, form the local pseudovariety
  $P_\Sigma$ of all $\Sigma$-generated finite $\D$-monoids $e\colon \Psi\Sigma^*\epito M$ with $M\in \mathscr{V}$. Then $F\Sigma$ is the $\Sigma$-generated profinite $\D$-monoid defined by $P_\Sigma$, see Remark~\ref{rem:pvartosmon} and Theorem~\ref{thm:identify-pseudovariety}.
\end{enumerate}
\end{rem}

Again, these constructions are mutually inverse:

\begin{theorem}\label{thm:reiterman} 
The maps $\mathcal{T}\mapsto \mathscr{V}$ and $\mathscr{V}\mapsto \mathcal{T}$ define a bijective correspondence between profinite equational theories and pseudovarieties of $\D$-monoids.
\end{theorem}

\begin{rem}
This theorem can be viewed as a categorical presentation of the well-known Reiterman-Banaschewski correspondence~\cite{Reiterman1982,Banaschewski1983}. The difference lies in the definition of a profinite theory: Reiterman and Banaschewski work with profinite equations (i.e.\ pairs of elements of free profinite monoids) while we work with quotients of free profinite monoids.  
\end{rem}

%

\section{Eilenberg-type Correspondences}
\label{sec:appliations}
Putting the results of our paper together we will now derive a number of Eilenberg-type theorems. Each of these theorems is an immediate consequence of the isomorphisms we established between our opfibrations $p$, $q$ and $q'$ (see the diagram in the Introduction) and the characterisation of their global sections. First, by Theorem~\ref{thm:identify-pseudovariety} we get another version of the
General Local Variety Theorem, i.e.\ Theorem~\ref{thm:local-eilenberg-thm}).
\begin{theorem}[General Local Variety Theorem II]
  There is a one-to-one correspondence between local varieties of languages
  over $\Sigma$ in~$\C$  and $\Sigma$-generated profinite $\D$-monoids:
  \[
    \FLang_\Sigma \cong \mathsf{Quo}(\widehat{\Psi\Sigma^*}).
  \]
\end{theorem}
Similarly, by Theorem~\ref{thm:fibrational-eilenberg-thm}, Corollary~\ref{cor:isoqq} and Theorem~\ref{thm:reiterman} we recover the main result of~\cite{Adamek2015}, where a completely different proof method was applied:

\begin{theorem}[General Variety Theorem]
  There is a one-to-one correspondence between varieties of languages in $\C$ and pseudovarieties of $\D$-monoids.
\end{theorem}

An interesting generalisation of this theorem emerges by restricting
$\Free(\CMon\D)$ to a subcategory. Recall that the pullback in
$\mathbf{Cat}$ of an opfibration $p\colon \E\to\B$ along any functor $F\colon \B'\to \B$
is again an opfibration, see e.g.,~\cite[Lemma 1.5.1]{Jacobs1999}.

\begin{definition}
  For a subcategory $\mathsf{C} \hookrightarrow \Free(\CMon\D)$, a
  \textbf{$\boldsymbol{\mathsf{C}}$-variety of languages in $\boldsymbol{\C}$}
  is a global section of the opfibration $p_\mathsf{C} \colon \FLang_\mathsf{C}
  \to \mathsf{C}$ obtained as the pullback of the opfibration $p$ along the
  inclusion. Similarly, a \textbf{profinite equational
    $\boldsymbol{\mathsf{C}}$-theory of $\boldsymbol{\D}$-monoids} is a global
  section of the opfibration $q_\mathsf{C}'\colon
  \FPFMon_\mathsf{C}\to\mathsf{C}$ obtained as the pullback of
  $q'\colon\FPFMon\to\Free(\CMon\D)$ along the inclusion.
  \[
    \xymatrix@C+1em{
      \FLang_\mathsf{C}\; \ar@{^(->}[r] \ar[d]_{p_\mathsf{C}}\pullbackcorner &
      \FLang \ar[d]^{p} \\
      \mathsf{C}\; \ar@{^(->}[r] & \Free(\CMon\D)
    }\qquad 
        \xymatrix@C+1em{
          \FPFMon_\mathsf{C}\; \ar@{^(->}[r] \ar[d]_{q_\mathsf{C}'}\pullbackcorner &
          \FPFMon \ar[d]^{q'} \\
          \mathsf{C}\; \ar@{^(->}[r] & \Free(\CMon\D)
        }
  \]
\end{definition}

More explicitly, a profinite equational $\mathsf{C}$-theory associates to each
$\Sigma\in \mathsf{C}$ a $\Sigma$-generated profinite monoid $e_\Sigma\colon
\widehat{\Psi\Sigma^*}\epito F\Sigma$ such that, for all $f\colon
\Psi\Sigma^*\to\Psi\Delta^*$ in $\mathsf{C}$, diagram
\eqref{eq:profinite-equational-theory} commutes for some $h$.  Similarly, a
$\mathsf{C}$-variety of languages determines a family
$(V_\Sigma)_{\Sigma\in\mathsf{C}}$, where $V_\Sigma$ is a local variety of
languages over $\Sigma$ in $\C$ and, for each $f\colon\Psi\Sigma^*\to\Psi\Delta^*$ in
$\mathsf{C}$, the local variety $V_\Sigma$ is closed under $f$-preimages of
languages in $V_\Delta$. For the case where $\C=\CBA$, $\D=\CSet$ and the
subcategory $\mathsf{C}$ contains all objects of $\Free(\CMon)$, this definition
coincides with the concept of a $\mathsf{C}$-variety of languages introduced by
Straubing~\cite{Straubing2002}. He also proved a special case of
Theorem~\ref{thm:cvar} below. Observe that since the opfibrations $p$ and $q'$
are isomorphic, so are their pullbacks $p_\mathsf{C}$ and $q_\mathsf{C}'$.
Therefore:

\begin{theorem}[General Variety Theorem for $\mathsf{C}$-varieties of languages]\label{thm:cvar}
  There is a one-to-one correspondence between $\mathsf{C}$-varieties of
  languages in $\C$ and profinite equational $\mathsf{C}$-theories of
  $\D$-monoids.
\end{theorem}

As an application of this theorem, let us choose $\mathsf{C}$ to be the full
subcategory of $\Free(\CMon\D)$ on a single object $\Sigma$. Then a
$\mathsf{C}$-variety of languages in $\C$ is precisely a local variety of
languages over $\Sigma$ in $\C$ closed under preimages of $\D$-monoid
endomorphisms $f\colon \Psi\Sigma^*\to\Psi\Sigma^*$.  We call such a local
variety \textbf{fully invariant}. A profinite equational $\mathsf{C}$-theory
consists of a single $\Sigma$-generated profinite $\D$-monoid $e\colon
\widehat{\Psi\Sigma^*}\epito F\Sigma$ such that, for all $\D$-monoid
endomorphisms $f\colon \Psi\Sigma^*\to\Psi\Sigma^*$, $e\widehat f$ factors
through $e$.
\[
        \xymatrix{
          \widehat{\Psi\Sigma^*} \ar[r]^{\widehat{f}} \ar@{->>}[d]_{e} &
          \widehat{\Psi\Sigma^*} \ar@{->>}[d]^{e} \\
          F\Sigma \ar@{-->}[r] & F\Sigma
        }
 \]
Again, such a $\Sigma$-generated profinite $\D$-monoid is called \textbf{fully
  invariant}. Hence full invariance means precisely that (in-)equalities are
stable under translations, i.e.\ for every $x, y \in \widehat{\Psi\Sigma^*}$ and
$f\colon\Psi\Sigma^*\to\Psi\Sigma^*$ we have that $e(x) = e(y)$ implies
$e(\widehat{f}x) = e(\widehat{f}y)$; in respect that $\mathscr{D}$-algebras are
ordered, $e(x) \leq e(y)$ implies $e(\widehat{f}x) \leq e(\widehat{f}y)$.
Therefore Theorem~\ref{thm:cvar} gives the following:

\begin{theorem}[Local Variety Theorem for Fully Invariant Varieties]
  There is a one-to-one correspondence between fully invariant local varieties
  over $\Sigma$ in $\C$ and fully invariant $\Sigma$-generated profinite $\D$-monoids. 
\end{theorem}
%

\section{Conclusions and Future Work}
\label{sec:conclusion}
In this paper we studied varieties of languages, pseudovarieties of monoids and profinite equational theories from an abstract fibrational viewpoint. This led us to conceptually new proofs and generalisations for a number of Eilenberg-Reiterman-type results.

Our notion of profinite equational theory is introduced on a rather abstract
level, and it would be helpful to characterise theories syntactically and compare them with classical developments~\cite{Reiterman1982,Banaschewski1983}. To this end one can observe that in the category of compact Hausdorff spaces every epimorphism is regular. Hence, if $\D$-algebras are non-ordered, every $\Sigma$-generated profinite $\D$-monoid
$e\colon\widehat{\Psi\Sigma^*}\twoheadrightarrow M$ is the coequaliser of its
kernel pair $\pi_1, \pi_2\colon E \rightrightarrows \widehat{\Psi\Sigma^*}$,
where $E$ is the kernel congruence defined by
\[
  E= \set{ (u, v) \in \widehat{\Psi\Sigma^*} \times
    \widehat{\Psi\Sigma^*}}{e(u) = e(v)}.
\]
Hence a profinite equational theory corresponds to a family of \emph{profinite
  equations}, i.e.\ pairs of elements of a free profinite monoid.  From this
observation it should be possible to obtain syntactic counterparts of our
results, e.g., a generalisation of the main result of Gehrke et~al.~\cite{Gehrke2008} that local varieties of languages in $\CBA$ and $\CDLat$ are definable by profinite identities.

In addition, it would be useful to develop a notion of \emph{morphism} between
profinite equational theories, and correspondingly between varieties of languages, hence lifting our generalised Eilenberg-Reiterman correspondences from an isomorphism of posets to an equivalence of categories. Such a result may further justify the importance of a categorical
treatment of algebraic automata theory.

\bibliographystyle{plain}

\clearpage
\appendix
\renewcommand\thetheorem{\thesection.\arabic{theorem}}

\section{Ind-completion and pro-completion}
\label{sec:ind-completion}
The following facts on ind/pro-completions are standard results,
see~\cite{Johnstone1982} for further detail.
\begin{definition}
\begin{enumerate}
\item An \textbf{ind-completion} of a small category $\A$ is a full and
  faithful functor $\A\rightarrowtail\Ind\A$ such that
  $\Ind\A$ has filtered colimits and every functor
  $F$ from $\A$ to a category $\B$ with filtered colimits has an
  extension $\overline{F}\colon \Ind A\to B$ which preserves filtered colimits
  and is unique up to natural isomorphism:
  \[
    \xymatrix{
      \A \ar@{ >->}[r] \ar[rd]_{F} & \Ind\A
      \ar@{-->}[d]^{\overline{F}}\\
      & \B
    }
  \]
  If $\A$ is finitely cocomplete, then $\Ind\A$ is complete
  and cocomplete. In particular, every locally finite variety $\mathscr{D}$ is an
  ind-completion of its full subcategory $\mathscr{D}_f$ on finite algebras.
\item Dually a \textbf{pro-completion} of a small category $\A$ is a full and
  faithful functor $\A\rightarrowtail\Pro\A$ such that
  $\Pro\A$ has cofiltered limits and every functor
  $F$ from $\A$ to a category $\B$ with cofiltered limits has an
  extension $\overline{F}\colon \Pro A\to B$ which preserves cofiltered limits
  and is unique up to natural isomorphism:
  \[
    \xymatrix{
      \A \ar@{ >->}[r] \ar[rd]_{F} & \Pro\A
      \ar@{-->}[d]^{\overline{F}}\\
      & \B
    }
  \]
  If $\A$ is finitely complete, then $\Pro\A$ is complete
  and cocomplete.
\end{enumerate}
\end{definition}

\begin{rem}A concrete construction of $\Ind\A$ is the following: let $\Ind\A$ be
  the full subcategory of the functor category $[\A^{op},\CSet]$ on all filtered
  colimits of representable functors $\A(\mathord{-},A)\colon \A^{op}\to\CSet$,
  and let $\mathcal{Y}\colon \A\to \Ind\A$ be the codomain restriction of the Yoneda embedding $A\mapsto \A(\mathord{-},A)$. Then $\mathcal{Y}$ is an ind-completion of $\A$. Analogously, one obtains the pro-completion as the dual Yoneda embedding $\mathcal{Y}^\op\colon \A \to \Pro\A$, $A\mapsto \A(A,\mathord{-})$. Note that  $\Pro\A = (\Ind\A^\op)^\op$.
\end{rem}

\begin{theorem}
  Given a small finitely complete and cocomplete category $\A$, there
  is an adjunction $F \dashv U \colon \Pro\A \to \Ind\A$ such
  that $\mathcal{Y}^\op = F \circ \mathcal{Y}$ and $\mathcal{Y} = U \circ
  \mathcal{Y}^\op$:
  \[
    \UseTwocells
    \xymatrix{
      & \A \ar[rd]^{\mathcal{Y}^\op} \ar[ld]_{\mathcal{Y}} \\
      \Ind\A \rrtwocell_U^F{'\bot} & & \Pro\A
    }
  \]
\end{theorem}
\begin{proof}
  $F$ and $U$ are the unique extensions of $\mathcal{Y}^\op$ and $\mathcal{Y}$ preserving filtered colimits and cofiltered limits, respectively.
  Since $\Ind\A$ consists of filtered colimits of representable functors
  $\mathcal{Y}(a) = \mathcal{A}(-, a)$ and similarly for $\Pro\A$, we
  have 
  \begin{align*}
    \Pro\A(F\Colim_a\mathcal{Y}a, \Lim_b \mathcal{Y}^\op b)
    & \cong \Lim_b \Pro\A(\Colim \mathcal{Y}^\op a, \mathcal{Y}^\op b) \\
    & \cong \Lim_a \Lim_b\Pro\A(\mathcal{Y}^\op a, \mathcal{Y}^\op b)
    \\
    & \cong \Lim_a \Lim_b \A(a, b) \\
    & \cong \Lim_a \Lim_b \Ind\A(\mathcal{Y}a, \mathcal{Y}b) \\
    & \cong \Lim_b \Ind\A(\Colim_a\mathcal{Y}a, \mathcal{Y}^\op b) \\
    & \cong \Ind\A(\Colim_a\mathcal{Y}a, U\Lim_b\mathcal{Y}^\op b).
  \end{align*}
\end{proof}

\section{Proofs}

\subsection*{Proof of Proposition~\ref{prop:characterisation_rhoT}}
Let $\CAut_0\Sigma$ and $\CAut_{0,lf}\Sigma$ denote the categories of $T_\Sigma^0$-coalgebras and locally finite $T_\Sigma^0$-coalgebras, respectively.
The functor $|T_\Sigma|\colon \C \to
\CSet$ is naturally isomorphic to $T_\Sigma^0 \circ |{-}|$, so the adjunction
$\Phi \dashv |{-}|\colon \C\to\CSet$ induces an adjunction $\CAut\Phi \dashv
\CAut|{-}|\colon\CAut\Sigma\to \CAut_0\Sigma$
by~\cite[Corollary~2.15]{Hermida1998}. The right adjoint~$\CAut|{-}|$ maps an
automaton~$(Q, \gamma)$ in $\C$ to its underlying automaton $(|Q|, |\gamma|)$ in $\CSet$,
and the left adjoint $\CAut\Phi$ maps an automaton $(Q_0, \gamma_0)$ in $\CSet$ to an
automaton in $\C$ with carrier $\Phi Q_0$.  Since  $\C$ is
locally finite, the adjunction restricts to one between 
the full subcategories $\CAut_{lf}$ and $\CAut_{0,lf}$ of locally finite
$\Sigma$-automata. Since the restricted right adjoint $\CAut|{-}|\colon \CAut_{lf}\Sigma \to \CAut_{0,lf}\Sigma$ preserves limits, it maps the terminal locally finite $T_\Sigma$-coalgeba $\rho T_\Sigma$ to the terminal locally finite $T_\Sigma^0$-coalgebra $\rho T_\Sigma^0$, i.e.\ to the automaton of regular languages.

\subsection*{Proof of Lemma~\ref{lem:hwelldef}}
\begin{enumerate}
\item For all $P\in \FVar_\Sigma$, the set $f_\sharp(P)$ forms a local
  pseudovariety of $\D$-monoids over $\Delta$. Indeed, closure under quotients
  is obvious. For closure under subdirect products let $e_{i}\colon \Psi\Delta^*\epito N_i$ ($i=1,2$) be two $\Delta$-generated $\D$-monoids in $f_\sharp(P)$, that is, $e_{i}f = g_ie_{M_i}$ for some $M_i\in P$ and morphisms $g_i$. We may assume that $M := M_1 = M_2$ -- otherwise replace $M_1$ and $M_2$ by their subdirect product $M_1\vee M_2$.  Hence the left diagram below commutes. By the fill-in property, there exists a
  unique morphism $h$ from $M$ to the subdirect product $N_1 \vee N_2$ of
  $N_1$ and $N_2$ such that the right diagram below commutes.
  \[
    \xymatrix@R+1em{
      \Psi\Sigma^* \ar@{->>}[d]_{e_M} \ar[r]^{f} & \Psi\Delta^* \ar[d]^{\<e_1, e_2\>}\\
      M \ar[r]_-{\<g_1, g_2\>}& N_1 \times N_2}
    \qquad\qquad
    \xymatrix@R+.5em{
      \Psi\Sigma^* \ar[r]^{f} \ar@{->>}[d]_{e_M} & \Psi\Delta^*
      \ar@{->>}[d]_{\mathrm{Im}\<e_1, e_2\>} \ar@/^1pc/[rd]^{\<e_1, e_2\>}\\
      M \ar@/_1pc/[rr]_{\<g_1, g_2\>} \ar@{-->}[r]^-{h} & N_1 \vee N_2 \ar@{ >->}[r] & N_1 \times N_2}
  \]
  Hence $N_1 \vee N_2$ lies in $f_\sharp(M)$.
\item $f_\sharp$ is clearly order-preserving, i.e.\ $P\subseteq P'$ implies $f_\sharp(P)\subseteq f_\sharp(P')$.
\item It remains to show the functoriality, i.e.\ $\id_\sharp = \id$ and $(g
f)_\sharp = g_\sharp f_\sharp$ for any two $\D$-monoid morphisms $f\colon
\Psi\Sigma^*\to \Psi\Delta^*$ and $g\colon\Psi\Delta^*\to\Psi\Gamma^*$. The first statement follows from the closure of local pseudovarieties under quotients. For the second one let $P\in \FVar_\Sigma$ and suppose that $K\in g_\sharp f_\sharp(P)$. Hence there exist finite $\D$-monoids $M\in P$ and $N\in f_\sharp(P)$ and $\D$-monoid morphisms making the diagram below commute.
  \[
    \xymatrix{
      \Psi\Sigma^* \ar[r]^{f}\ar@{->>}[d] & \Psi\Delta^* \ar[r]^{g} \ar@{->>}[d]
      & \Psi\Gamma^* \ar@{->>}[d] \\ M \ar[r] & N \ar[r] & K
    }
  \] 
  This implies $K\in (gf)_\sharp(P)$.  On the other hand, suppose that $K \in
  (gf)_\sharp(P)$, i.e.\ there exists some $M\in P$ and a $\D$-monoid morphism $h\colon M \to K$ such that
  the left diagram below commutes. Consider the factorisation of
  $e_K\circ g\colon\Psi\Delta^* \rightarrow\Psi\Gamma^* \twoheadrightarrow K$
  in the right diagram:
  \[
    \xymatrix{
      \Psi\Sigma^* \ar[r]^{f}\ar@{->>}[d] & \Psi\Delta^* \ar[r]^{g}
      & \Psi\Gamma^* \ar@{->>}[d]^{e_K} \\
      M \ar[rr]_{h} &  & K
    }
    \qquad\qquad
    \xymatrix{
      \Psi\Sigma^* \ar[r]^{f}\ar@{->>}[d] & \Psi\Delta^* \ar[r]^{g} \ar@{->>}[d]
      & \Psi\Gamma^* \ar@{->>}[d]^{e_K} \\
      M \ar@{-->}[r] \ar@/_1pc/[rr]_h & N \ar@{ >->}[r] & K
    }
  \]
  By the fill-in property $h$ factors through the submonoid $N$ of the finite
  monoid $K$. Hence $N\in f_\sharp(P)$ and
  $K\in (g_\sharp f_\sharp)(P)$. \qedhere
\end{enumerate}

\subsection*{Proof of Theorem~\ref{thm:identify-pseudovariety}}

\begin{lemma}
  \label{lem:canonical_limit}
  Every profinite $\D$-monoid is the cofiltered limit of its finite quotients.
\end{lemma}
\begin{proof}
  Since the category $\CPFMon\D$ is the pro-completion of finite $\D$-monoids,
  every profinite $\D$-monoid is the limit of its canonical
  cofiltered diagram
  \[
     \xymatrix{(M\downarrow\CMon_f\D)\ar[r]^-{Q}&
       \CMon_f\D \ar@{ >->}[r]^-{i} &\CPFMon\D}
  \]
  where $(M \downarrow \CMon_f\D)$ is the comma category from $M$ to the
  category of finite $\D$-monoids, and $Q$ is the projection functor. However, given this
  canonical diagram, we can always factor every morphism $M \to N$ for $N \in
  \CMon_f\D$ into a surjective morphism and an embedding:
  \[
    \xymatrix@C-1em{
      & & M \ar@/^3pc/[rrdd]^g \ar@{->>}[rd]^{\mathrm{Im}(g)}
      \ar@/_3pc/[lldd]_f \ar@{->>}[ld]_{\mathrm{Im}(f)} \\
      & f[M] \ar@{ >->}[ld] \ar@{-->}[rr] & & g[M] \ar@{ >->}[rd] \\
     N \ar[rrrr]_h & & & & N' 
    }
  \]
  This diagram consisting of all finite quotients is also cofiltered, since $i\circ Q$ is
  cofilterd.  Then, it is easy to see that $M$ with $\{ \mathrm{Im}(f)\colon M
  \to f[M]\}_{f \in (M\downarrow\CMon_f\D)}$ is a cofiltered limit.
\end{proof}
\begin{proof}[Proof of Theorem~\ref{thm:identify-pseudovariety}]
 
 (a) Let $e_M\colon \widehat{\Psi\Sigma^*} \twoheadrightarrow M$ be a
  profinite $\Sigma$-generated $\D$-monoid, and suppose that $K$
  is a finite quotient of $M$. Note that $K$ is finitely copresentable in
  $\CPFMon\D$, so since $\widehat{\Psi\Sigma^*}$ is the limit of all finite quotients of $\Psi\Sigma^*$, we see that $ee_M$ factors through some limit projection
  $\pi_N$:
   \[
     \xymatrix{
       \Psi\Sigma^* \ar[r] \ar@{->>}[rd] & \widehat{\Psi\Sigma^*} \ar@{->>}[d]_{\pi_N}
       \ar@{->>}[r]^{e_M} & M\ar@{->>}[d]^{e} \\ & N \ar@{->>}[r]_f & K
     }
   \]
 Therefore $K$ is a $\Sigma$-generated $\D$-monoid. It now immediately follows that the set $\mathscr{V}_M$ of finite quotients of $M$ forms  a local pseudovariety over $\Sigma$.  Clearly, the construction $M \mapsto \mathscr{V}_M$ is order-preserving, and it is injective by Lemma~\ref{lem:canonical_limit}.

 (b) Conversely, we can view every local pseudovariety $P\in\FVar_\Sigma$ as a diagram in $\CPFMon\D$ defined by 
   \[
     \xymatrix{P\;\ar@{^(->}[r]^-{i_P}&\Quo(\Psi\Sigma^*)\ar[r]^-{Q}&\CPFMon\D}
   \]
  where $i_P$ is the full inclusion and $Q$ is the projection functor mapping
  $\Psi\Sigma^* \twoheadrightarrow M$ to $M$ and $f\colon M \twoheadrightarrow
  M'$ in $\Quo(\Psi\Sigma^*)$ to $f$. 

  Note that each $M \in P$ with the discrete topology is a
  non-empty compact Hausdorff space. Then $M_P\defeq \Lim (Q \circ i_P)$ is a profinite $\D$-monoid where
  each limit projection $\pi_M$ is surjective by
  Lemma~\ref{lem:facts-inverse-systems}.  Suppose that $P \subseteq P'$.  Then
  there exists a mediating morphism from $M_{P'}$ to $M_{P}$, since the 
 projections $M_{P'} \xrightarrow{\pi_M} M$ for $M \in P$ form  a cone
  over $Q\circ i_{P}$. This mediating morphism is surjective, because every $\pi_M$ is surjective. In particular, taking $P'=\Quo(\Psi \Sigma^*)$ we get a surjective morphism $\widehat{\Psi\Sigma^*}\epito M_P$. (Recall that $\widehat{\Psi\Sigma^*}$ is by definition the limit of all finite
  quotients of $\Psi\Sigma^*$.) 

  (c) To show that the two construction of (a) and (b) are mutually inverse, we need to prove that, given $P\in \FVar_\Sigma$, every finite quotient $e_M\colon M_P = \Lim (i_P
  \circ Q) \twoheadrightarrow M$ is contained in $P$. Since $M$ is finitely
  copresentable, the morphism $e_M$ factors through some $N \in P$, so $M$ must be a
  quotient of $N$; that is, $M \in P$. We conclude the construction 
  $P\mapsto M_P$ is surjective. It is also order-preserving by the argument given in (b).
  
 (d)  The second part of theorem follows by a straighforward use of universal
  properties. 
\end{proof}

\subsection*{Proof of Theorem~\ref{thm:reiterman}}

The proof proceeds through several lemmas.

\begin{lemma}
  Given a profinite equational theory $\mathcal{T}$ of $\D$-monoids, the class $\mathscr{V}$ associated to $\mathcal{T}$ forms a pseudovariety of $\D$-monoids.
\end{lemma}

\begin{proof}

We need to show closure under quotients, submonoids and finite products. To this end, let $M \in \mathscr{V}$ and also finitely many $M_i \in \mathscr{V}$ be given. In the first two cases below, $f$ refers to a morphism from
  $\widehat{\Psi\Sigma^*}$ to a quotient and a submonoid of $M$ respectively.
  For the last case, $f$ is a morphism to the finite product $\prod_i M_i$. See
  following diagrams for references.
  \begin{figure}[h!]
  \renewcommand{\figurename}{Diagram}
    \centering
    \begin{minipage}{0.3\textwidth}
      \[
        \xymatrix{
          {\widehat{\Psi\Sigma^*}}\ar@{->>}[r]^{e_\Sigma}  \ar[d]_{f} \ar[rd]|h
          & F\Sigma \ar[d]^{\overline{h}} \\
          N & M \ar@{->>}[l]^{e}
        }
      \]
      \caption{Quotients}
      \label{dig:quotients}
    \end{minipage}%
    \begin{minipage}{0.3\textwidth}
      \[
        \xymatrix{
          {\widehat{\Psi\Sigma^*}} \ar@{->>}[r]^{e_\Sigma} \ar[d]_{f} &
          F\Sigma\ar[d]^{\overline{mf}}\ar@{-->}[ld]|h \\
            N \ar@{>->}[r]_{m} & M
        }
      \]
      \caption{Submonoids}
      \label{dig:submonoids}
    \end{minipage}
    \begin{minipage}{0.3\textwidth}
      \[
        \xymatrix{
          {\widehat{\Psi\Sigma^*}} \ar@{->>}[r]^{e_\Sigma} \ar[d]_{f} & F\Sigma
          \ar[d]^{\overline{\pi_if}} \ar@{-->}[ld]|h \\
          {\prod_i M_i} \ar[r]_{\pi_i} & M_i
        }
      \]
      \caption{Finite products}
      \label{dig:finite-products}
    \end{minipage}
  \end{figure}
  \begin{description}
    \item[Quotients:]
      Given a quotient $N$ of $M$ with $e\colon M \twoheadrightarrow N$,
      since free algebras $\widehat{\Psi\Sigma^*}$ are projective there
      exists $h$ with $f = eh$. By assumption $h$ factors through $e_\Sigma$ via some $\overline h$. Hence $f$ factors through $e_\Sigma$ via $e\overline h$.
    \item[Submonoids:]
      Given a submonoid $N$ of $M$, the composite $mf$ factors through
      $e_\Sigma$ by assumption. By the fill-in property, there is a
      morphism $h\colon F\Sigma \to N$ such that Diagram~\ref{dig:submonoids} commutes.
    \item[Finite products:]
      Every $\pi_i f$ factors through $e_\Sigma$ by assumption, so
      there is a mediating morphism $h\defeq \< \overline{\pi_if}\>$ such that
      Diagram~\ref{dig:finite-products} commutes.\qedhere
  \end{description}
\end{proof}

\begin{lemma}
  Given a pseudovariety $\mathscr{V}$ of $\D$-monoids the corresponding
  morphisms $e_\Sigma\colon\widehat{\Psi\Sigma^*} \epito F\Sigma$ form a profinite
  equational theory.
\end{lemma}
\begin{proof}
  Recall that $P_\Sigma$ is the set of $\Sigma$-generated monoids
  in~$\mathscr{V}$. Since $\mathscr{V}$ is a pseudovariety,
  $P_\Sigma$ is closed under quotients and subdirect products, so $P_\Sigma$ is
  a local pseudovariety over $\Sigma$.  corresponding uniquely to a quotient
  $e_\Sigma\colon\widehat{\Psi\Sigma^*} \twoheadrightarrow F\Sigma$ of the free
  profinite monoid. To see that the morphisms $e_\Sigma$ form a profinite equational theory, use Theorem~\ref{thm:identify-pseudovariety}: 
  for every $f\colon \Psi\Sigma^* \to \Psi\Delta^*$ and every
  $e\colon \Psi\Delta^* \twoheadrightarrow N$ in $P_\Delta$, the factorisation
  $\xymatrix@C+1em{\Psi\Sigma^* \ar@{->>}[r]^-{\mathrm{Im}(ef)} & M \ar@{ >->}[r] & N}$ of
  $ef$ fulfils the left-hand diagram in the Theorem 
  where the $\Sigma$-generated monoid $M$ of $N$ is in $P_\Sigma$ by the fact that
  $\mathscr{V}$ is closed under submonoids.  Hence the right diagram in the Theorem also commutes for some $h$, so it follows that the collection
  $\{e_\Sigma \}_\Sigma$ forms a profinite equational theory.
\end{proof}
Using the following lemma a straightforward verification shows that the constructions $\mathcal{T}\mapsto \mathscr{V}$ and $\mathscr{V}\mapsto\mathcal{T}$ are mutually inverse.
\begin{lemma}
  Let $\mathscr{V}$ be the pseudovariety corresponding to a profinite theory
  $(e_\Sigma\colon \widehat{\Psi\Sigma^*}\epito F\Sigma)_\Sigma$.
  Then $M \in \mathscr{V}$ if and only if $M$ is a quotient of~$F|M|$.
\end{lemma}
\begin{proof}
  Suppose that $M \in\mathscr{V}$. Then $M$ is a quotient of the free
  $\D$-monoid $\Psi|M|^*$ generated by $M$ itself, so it is also a
  quotient of the free profinite $\D$-monoid $\widehat{\Psi|M|^*}$. By
  assumption, the quotient map $\widehat{\Psi|M|^*}\twoheadrightarrow M$
  factors though $F|M|$ via some morphism that is necessarily surjective. 
  The other direction follows from the projectivity of $\widehat{\Psi|M|^*}$. 
\end{proof}

\end{document}